\documentclass[conference]{IEEEtran}
\usepackage[cmex10]{amsmath}

\usepackage{rays_defs}

\usepackage{xfrac}

\usepackage{pgf,tikz}
\usepackage{pgfplots}

\usepackage{url}

\usepackage{rays_defs}
\usepackage{sys_ind_defs}

\renewcommand{\baselinestretch}{0.99}\small\normalsize

\begin{document}

\title{Compressive Identification of Linear Operators}

\author{\IEEEauthorblockN{Reinhard Heckel and Helmut B\"olcskei}
\IEEEauthorblockA{Department of Information Technology and Electrical Engineering, ETH Zurich\\
E-mail: \{heckel,boelcskei\}@nari.ee.ethz.ch}}

\maketitle

\renewcommand{\sfrac}[2]{#1/#2}

\begin{abstract}
We consider the problem of identifying a linear deterministic operator from an input-output measurement. For the large class of continuous (and hence bounded) operators, under additional mild restrictions, we show that stable identifiability is possible if the total support area of the operator's spreading function satisfies $\SSF \leq \sfrac{1}{2}$. This result holds for arbitrary (possibly fragmented) support regions of the spreading function, does not impose limitations on the total extent of the support region, and, most importantly, does not require the support region of the spreading function to be known prior to identification. 
Furthermore, we prove that asking for identifiability of only almost all operators, stable identifiability is possible if $\SSF \leq 1$. 
This result is surprising as it says that there is no penalty for not knowing the support region of the spreading function prior to identification. 
\end{abstract}

\IEEEpeerreviewmaketitle

\section{Introduction}
Identification of deterministic linear operators from an input-output measurement is an important problem in many fields of engineering. Concrete examples include system identification in control theory and practice and the measurement of wireless communication channels. 

It is natural to ask under which conditions on the operator, identification from an input-output measurement is possible in principle, and how one would go about choosing the probing signal and extracting the operator from the corresponding output signal. This paper addresses these questions in generality by considering the (large) class of linear operators that can be represented as a continuous weighted superposition of time-frequency shift operators, i.e., the operator's response to the signal $\prsig(t)$ is given by
\begin{equation}
y(t) = \int_\tau \int_\nu  \sfunc (\tau,\nu) \prsig(t-\tau)  e^{j 2\pi \nu t}  d\nu d\tau
\label{eq:ltvsys}
\end{equation}
where $\sfunc$ denotes the spreading function of the operator  $H$. 
The representation theorem \cite[Thm. 14.3.5]{groechenig_foundations_2001} states that the action of any continuous (and hence bounded) linear operator, under additional mild restrictions, can be represented as in \eqref{eq:ltvsys}. 
In the communications literature \cite{kailath_measurements_1962,bello_measurement_1969,bello_characterization_1963}, operators with input-output relation as in \eqref{eq:ltvsys} are referred to as linear time-varying channels/systems and $\sfunc$ is the delay-Doppler spreading function. 
For the special case of linear time-invariant (LTI) systems, we have $\sfunc(\tau,\nu) = h(\tau)\delta (\nu)$, so that \eqref{eq:ltvsys} reduces to 
\begin{equation}
y(t) = \int_\tau h(\tau) \prsig (t-\tau) d\tau.
\label{eq:lti_system}
\end{equation}
The question of identifiability of LTI systems is readily answered by noting that the system's response to $\prsig(t) = \delta(t)$ is given by its impulse response $h(t)$, which fully characterizes the input-output relation according to \eqref{eq:lti_system}. LTI systems are therefore always identifiable, provided that the input signal can have infinite bandwidth and we can observe the output signal over an infinite duration. 
In the general case (i.e., for LTV systems), the situation is fundamentally different. Specifically, 
 Kailath's landmark paper  \cite{kailath_measurements_1962} shows that LTV systems with spreading function compactly supported on a rectangle are identifiable if and only if  the spreading function's support area satisfies $\SSF \leq 1$. 
Bello \cite{bello_measurement_1969} later pointed out that  Kailath's identifiability result continues to hold, even if the support region of $\sfunc$ is scattered across the $(\tau,\nu)$-plane, as long as the total support area satisfies $\SSF \leq 1$. 
Kozek and Pfander \cite{kozek_identification_2005} and Pfander and Walnut \cite{pfander_measurement_2006} provided functional-analytic proofs of the results in \cite{kailath_measurements_1962,bello_measurement_1969}. 
Common to \cite{kailath_measurements_1962,bello_measurement_1969,kozek_identification_2005,pfander_measurement_2006} is the assumption of the support region of  $\sfunc$ being known prior to identification. This is clearly restrictive and often impossible to realize in practice. 

\paragraph*{Contributions}
In this paper, we consider the problem of identifying deterministic linear operators with spreading function compactly supported on an \emph{unknown}, possibly fragmented, region in the $(\tau,\nu)$-plane. 
We do not impose limitations on the total extent of the support region. Our main result shows that an operator with input-output relation \eqref{eq:ltvsys} is identifiable, without prior knowledge of the spreading function's support region, if and only if the total support area of $\sfunc$ satisfies $\SSF  \leq \sfrac{1}{2}$. 
We then show that this factor of two penalty---compared to the case where the support region is known in advance \cite{kailath_measurements_1962,bello_measurement_1969}---can be eliminated, 
if one asks for identifiability of \emph{almost all} operators only. This result is surprising as it shows that (for almost all operators) there is no price to be paid for not knowing the spreading function's support region in advance. 
We discuss the design of probing signals and we outline an algorithm which, in the noiseless case, provably recovers all operators with $\SSF \leq \sfrac{1}{2}$. Furthermore, we present an algorithm which, again in the noiseless case, provably recovers almost all operators with $\SSF \leq 1$. 

\paragraph*{Relation to previous work}
Recently Bajwa et al. \cite{bajwa_identification_2010} considered the identification of LTV-systems with 
spreading function $\sfunc$ supported  in a rectangle of area $\SSF \leq 1$ and consisting of a finite number of discrete components with unknown delays and Doppler shifts. 
In the present paper, we allow general spreading functions that can be supported in the entire $(\tau,\nu)$-plane with possibly fragmented support region. 
\paragraph*{Notation} We use lowercase boldface letters to denote column vectors, e.g., $\mbf x$, and uppercase boldface letters to designate matrices, e.g., $\mbf X$. 
The superscripts $\conj{}$, $\herm{}$, and $\transp{}$ stand for complex conjugation, Hermitian transposition, and transposition, respectively.  The space spanned by the columns of $\mbf X$ is designated by $\range(\mbf X)$. The entry in the $k$th row and $l$th column of $\mbf X$ is denoted by $[\mbf X]_{k,l}$. For the vector $\mbf x$, the Euclidean norm is denoted as $\norm[\ell_2]{\mbf x}$ and the $k$th entry of $\mbf x$ is $[\mbf x]_k$. 
$|\Omega|$ stands for the cardinality of the set $\Omega$. For two sets $\Omega_1$ and $\Omega_2$ we define set addition as $\Omega_1+\Omega_2= \{\omega: \omega_1 +\omega_2, \, \omega_1 \in \Omega_1, \omega_2 \in \Omega_2\}$. $\delta(x)$ denotes the Dirac delta function. 
For a function $f(\mbf x)$, $\supp (f)$ denotes the support set of $f$.  For two functions $f(\mbf x), g(\mbf x)$, defined on $\Omega$, we write $\innerprod[]{f}{g} \defeq \int_\Omega f(\mbf x) \conj{g}\!(\mbf x) d \mbf x$ for the inner product; $\norm{f} \defeq \sqrt{ \innerprod[]{f}{f} }$ is the norm of $f$. 

\vspace{-0.1cm}
\section{Problem Formulation \label{sec:prform}}
\vspace{-0.05cm}

Given normed linear spaces\footnote{
Since our identifiability proof relies on sending Dirac delta impulses, we need to choose $\opfrom$ such that it contains generalized functions. To keep the exposition simple, we will not dwell on the resulting functional-analytic subtleties. Instead, we refer the interested reader to \cite{kozek_identification_2005,pfander_measurement_2006} for a description of the rigorous mathematical setup required here.  
} \!\!\!\! $\opfrom, \opto$, we consider linear operators $H: \opfrom \mapsto \opto$ that can be represented as a continuous weighted superposition of translation operators $T_\tau$, with $(T_{\tau} \prsig )(t) \defeq \prsig (t-\tau)$, and modulation operators $M_{\nu}$ with $(M_\nu \prsig)(t) \defeq e^{j2\pi \nu t} \prsig (t)$:
\begin{equation}
\vspace{-0.15cm}
(H\prsig)(t) \defeq \int_\tau \int_\nu  \sfunc (\tau,\nu) ( M_{\nu}  T_{\tau} \prsig )(t)   d\nu \hspace{0.02cm}  d\tau. 
\label{eq:hilbertschmidt_s}
\end{equation}
This is a rather general setting, since according to \cite[Thm. 14.3.5]{groechenig_foundations_2001}, any continuous (and hence bounded) linear operator, under additional mild restrictions, can be represented as in \eqref{eq:hilbertschmidt_s}. In the following, denote the linear space of operators that can be represented according to \eqref{eq:hilbertschmidt_s} by $\mc H$, and define the inner product on this space by  
\vspace{-0.05cm}
\[
\innerprod[ \opclass]{H_1}{H_2} \defeq \innerprod[]{s_{\hspace{-0.02cm}H_1}}{s_{\hspace{-0.02cm}H_2}} = \int_{\mb R^2}    s_{\hspace{-0.02cm}H_1}(\tau,\nu) \conj{s}_{\hspace{-0.02cm}H_2}(\tau,\nu) d(\tau,\nu),      
 \vspace{-0.1cm}
\]
$H_1,H_2 \in \mc H $, with the induced norm $\norm[\opclass]{H} \defeq \sqrt{\innerprod[ \opclass]{H}{H}}$. 

\subsubsection*{Restrictions on the spreading function} We assume that the support region of $\sfunc$ has the form 
\begin{equation}
M_\Gamma  \defeq \!\! \bigcup_{(k,m) \in \Gamma } \left( \Un + \left(k \csamp ,  \frac{m}{\csamp L} \right) \right)  \subseteq [0, \tau_{\max}) \times [0, \nu_{\max})
\label{eq:suppgridcont}
\vspace{-0.1cm}
\end{equation}
where $\Un \defeq \left[ 0, \csamp  \right) \times \left[0, 1/\csamp L  \right)$ is a ``fundamental cell'' in the $(\tau,\nu)$-plane  
and $\csamp \in \mb R$ is a parameter whose role will become clear shortly. The set of ``active cells'' is specified by $\Gamma  \subseteq \Sigma \defeq \{(0,0), (0,1),...,(L-1,L-1)\}$. Consequently, we have $\tau_{\max} = \csamp L$ and $\nu_{\max} = 1/\csamp$. 
We denote the support area of  $M_\Gamma$ as $\area(M_\Gamma)$. 
 Given a general support region for $\sfunc$, possibly fragmented and spread over the entire $(\tau,\nu)$-plane, we can choose $\csamp$ and $L$ such that this region can  be approximated arbitrarily well by a support region of the form \eqref{eq:suppgridcont} (see Fig. \ref{fig:supex}).  
\begin{figure}[h!]
\begin{center}

\begin{tikzpicture}[scale=0.75] 

\foreach \y in {1,1.2}  \filldraw[fill=black!10,draw=black!10] (1.8,\y) rectangle (1.8+0.2,\y+0.2); 
\foreach \y in {0.8,1,...,1.6}  \filldraw[fill=black!10,draw=black!10] (2,\y) rectangle (2+0.2,\y+0.2); 
\foreach \y in {0.8,1,...,2}  	\filldraw[fill=black!10,draw=black!10] (2.2,\y) rectangle (2.2+0.2,\y+0.2); 
\foreach \y in {0.8,1,...,2.4} \filldraw[fill=black!10,draw=black!10] (2.4,\y) rectangle (2.4+0.2,\y+0.2); 
\foreach \y in {0.8,1,...,2.6} \filldraw[fill=black!10,draw=black!10] (2.6,\y) rectangle (2.6+0.2,\y+0.2); 
\foreach \y in {0.8,1,...,3.2} \filldraw[fill=black!10,draw=black!10] (2.8,\y) rectangle (2.8+0.2,\y+0.2); 
\foreach \y in {0.8,1,...,3.6} \filldraw[fill=black!10,draw=black!10] (3  ,\y) rectangle (3   +0.2,\y+0.2); 
\foreach \y in {0.8,1,...,3.8} \filldraw[fill=black!10,draw=black!10] (3.2,\y) rectangle (3.2+0.2,\y+0.2); 
\foreach \y in {0.8,1,...,3.8} \filldraw[fill=black!10,draw=black!10] (3.4,\y) rectangle (3.4+0.2,\y+0.2); 
\foreach \y in {0.8,1,...,3.6}    	\filldraw[fill=black!10,draw=black!10] (3.6,\y) rectangle (3.6+0.2,\y+0.2); 
\foreach \y in {0.8,1,...,2.8} \filldraw[fill=black!10,draw=black!10] (3.8,\y) rectangle (3.8+0.2,\y+0.2); 

\foreach \x in {0.4,0.6,...,1.6}  \filldraw[fill=black!10,draw=black!10] (\x,4) rectangle (\x+0.2,4+0.2); 
\foreach \x in {0.4,0.6,...,1.6}  \filldraw[fill=black!10,draw=black!10] (\x,4.2) rectangle (\x+0.2,4.2+0.2);
\foreach \x in {0.4,0.6,...,1.4}  \filldraw[fill=black!10,draw=black!10] (\x,4.4) rectangle (\x+0.2,4.4+0.2);
\foreach \x in {0.4,0.6,...,1.2}  \filldraw[fill=black!10,draw=black!10] (\x,4.6) rectangle (\x+0.2,4.6+0.2);
\foreach \x in {0.4,0.6,...,1}  \filldraw[fill=black!10,draw=black!10] (\x,4.8) rectangle (\x+0.2,4.8+0.2);


\draw[help lines, step=0.2] (0,0) grid (5,5); 

\pgfplothandlerclosedcurve 
\pgfplotstreamstart 
\pgfplotstreampoint{\pgfpoint{0.52cm}{4.1cm}} 
\pgfplotstreampoint{\pgfpoint{0.52cm}{5cm}} 
\pgfplotstreampoint{\pgfpoint{1.6cm}{4.1cm}} 
\pgfplotstreamend \pgfusepath{stroke}

\pgfplothandlerclosedcurve 
\pgfplotstreamstart 
\pgfplotstreampoint{\pgfpoint{2.5cm}{2.3cm}} 
\pgfplotstreampoint{\pgfpoint{3.5cm}{4cm}} 
\pgfplotstreampoint{\pgfpoint{3.95cm}{1.1cm}} 
\pgfplotstreampoint{\pgfpoint{2cm}{1cm}} 
\pgfplotstreamend \pgfusepath{stroke}

\draw[->] (0,0) -- (5.5,0); 
\draw[->] (0,0) -- (0,5.5);
\foreach \x in {0,0.2,...,5} \draw (\x ,-1pt) -- (\x ,1pt);
\foreach \y in {0,0.2,...,5} \draw (-1pt,\y) -- (1pt,\y);
\node at (0.2,0) [anchor=north] {$\csamp$};
\node at (0,0.2) [anchor=east] {$\frac{1}{\csamp L}$};
\node at (5,0) [anchor=north] {$\csamp L$};
\node at (0,5) [anchor=east] {$\frac{1}{\csamp}$};
\node at (2.5,-0.5) [anchor=north] {$\tau$};
\node at (-0.5,2.5) [anchor=south, rotate=90] {$\nu$};
\end{tikzpicture}\end{center}
\vspace{-0.5cm}
\caption{
Approximation of a general support region. } 
\label{fig:supex}

\vspace{-0.5cm}

\end{figure}
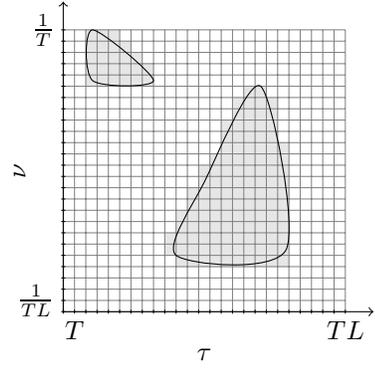

\vspace{-0.09cm}
\subsection{Identifiability \label{sec:identf}}

\newcommand{\opclassgen}[0]{\ensuremath{ \mc Q}}

Let us next formally define  the notion of identifiability of a set of operators $\opclassgen \subseteq \opclass$. The set $\opclassgen$ is said to be identifiable from an input-output measurement, if there exists a probing signal $\prsig \in \opfrom$ such that for each $H\in \opclassgen$, the action of the operator on the probing signal, $H\prsig$, uniquely determines $H$, i.e., if there exists an $\prsig \in \opfrom$ such that 
\vspace{-0.1cm}
\begin{equation}
H_1 \prsig = H_2 \prsig \implies H_1=H_2, \quad \forall \hspace{0.1cm} H_1,H_2 \in \opclassgen.
\label{cond:onetoone}
\end{equation}
Identifiability is hence equivalent to invertibility of the mapping
\begin{equation}
\opclassgen \to \opto: H \mapsto H \prsig
\label{eq:mapping} 
\end{equation}   
induced by the probing signal $\prsig$. 
Invertibility alone is typically not sufficient as one wants to recover $H$ from $H \prsig$ in a numerically stable fashion, i.e., we want small errors in $H \prsig$ to result in small errors in the identified operator. 
This requirement implies that the inverse of the mapping \eqref{eq:mapping} must be continuous (and hence bounded), which finally motivates the following definition. 
\begin{definition}
\label{def:defstableidentification}
We say that $\prsig$ {\it stably} identifies $\opclassgen$ if there exist constants $0 < \alpha \leq \beta < \infty$ such that for all pairs $H_1,H_2 \in \opclassgen$,
\begin{equation}
\alpha \norm[\opclass]{H_1 -H_2} \leq \norm[]{ H_1\prsig - H_2 \prsig  } \leq \beta \norm[\opclass]{ H_1-H_2 }.
\label{eq:cond_stable}
\vspace{-0.065cm}
\end{equation}
Furthermore, we say that $\opclassgen$ is stably identifiable, if there exists an $\prsig \in \opfrom$ such that $\prsig$ stably identifies $\opclassgen$. 
\end{definition}

Note that \eqref{eq:cond_stable} is stated in terms of differences of operators, since $(H_1-H_2)$ is not necessarily contained in the set $\mc Q$. 
The lower bound in \eqref{eq:cond_stable} guarantees that the inverse of \eqref{eq:mapping} exists and is bounded and hence continuous, as desired. 
Proving that $\prsig$ stably identifies $\opclassgen$ essentially amounts to proving that $\alpha>0$ in Definition \ref{def:defstableidentification}.  
The ratio $\beta/\alpha$ provides a measure for the noise sensitivity of the identification process, the closer $\beta/\alpha$  to one, the better. 
\vspace{-0.1cm}
\section{Main Results}
Before stating our main results, 
we define the set of operators with $\sfunc$ supported on a given area $M_\Gamma$:
\begin{equation}
\opclassr  \defeq \{ H \in \mc H : \supp(\sfunc) \subseteq M_{\Gamma} \}.
\end{equation}
In this notation, Kailath \cite{kailath_measurements_1962} discussed the case where $M_\Gamma$ is a (single) rectangle, and Bello \cite{bello_measurement_1969} analyzed  the case where $M_\Gamma$ is an arbitrary region, possibly fragmented and spread over the entire $(\tau,\nu)$-plane. We first recall the key results in \cite{bello_measurement_1969,pfander_measurement_2006} on the identification of $\opclassr$ under the assumption of $M_\Gamma$ known: 
\begin{theorem}[\cite{bello_measurement_1969,pfander_measurement_2006}]
Let $M_\Gamma$ be given. 
The set of operators $\opclassr$ is stably identifiable if and only if $\area(M_\Gamma) \leq 1$.
\label{pr:necessity}
\label{pro:knownsupport}
\end{theorem}
In the following, we consider the set of operators 
\begin{equation}
\opclassbr \defeq \bigcup_{ \area (M_\Gamma) \leq \SSF} \mc H_{M_\Gamma}
\label{eq:blindop}
\end{equation}
 which consists of {\it all} sets $\opclassr$ such that $\area (M_\Gamma) \leq \SSF$. 
 
 Our main results are as follows. 
\begin{theorem}
The set of operators  $\mc X(\SSF)$ is stably identifiable if and only if 
$\SSF \leq \sfrac{1}{2}$.
\label{pr:identifiabilitynecsparse}
\label{th:sufficiency}
\end{theorem}

Note that Theorem \ref{th:sufficiency} applies to the union of \emph{all} sets of operators with spreading function supported on a region with area no larger than $\SSF$, and, in particular, does not need the support region of $\sfunc$ to be known in advance. The main implication of Theorem \ref{th:sufficiency} is that the penalty for not knowing the support region of $\sfunc$ in advance is a factor of two in the support set size of the spreading function. 
We can eliminate this penalty by relaxing the identification requirement to \emph{almost all} $H\in \opclassbr$.

\begin{theorem}
Almost all $H \in \opclassbr$ can be stably identified if $\SSF < 1$.  
\label{th:almost_all}
\end{theorem}

The factor of two penalty in Theorem \ref{th:sufficiency} has the same roots as the factor of two penalty in spectrum-blind sampling \cite{feng_spectrum-blind_1996,bresler_spectrum-blind_2008,feng_universal_1997,mishali_blind_2009}, sparse signal recovery \cite{donoho_optimally_2003}, and in the recovery of signals that lie in a union of subspaces \cite{lu_theory_2008}. 
It was recognized before---in the context of spectrum-blind sampling---that the factor of two penalty can be eliminated by relaxing the recovery requirement to \emph{almost all} signals \cite{feng_spectrum-blind_1996,bresler_spectrum-blind_2008,feng_universal_1997}.

\section{Necessity in Theorem \ref{pr:identifiabilitynecsparse} \label{sec:necessary}}

To prove necessity in Theorem \ref{pr:identifiabilitynecsparse} we start with the following lemma which states an equivalent condition on stable identifiability of $\opclassbr$. This condition is often easier to verify than that in Definition \ref{def:defstableidentification}. 

\begin{lemma}
$\prsig$ stably identifies $\opclassbr$ if and only if it stably identifies {\it all} sets
\[
\mc H_{M_\Phi \cup M_\Theta} \defeq \{H: H = H_1 - H_2 , H_1 \in \mc H_{M_\Phi}, H_2 \in \mc H_{M_\Theta} \}
\]
 with $\area(M_\Phi) \leq \SSF$ and $\area(M_\Theta) \leq \SSF$ where $\Phi, \Theta \in \Sigma$.  
\label{pr:blind_identification}
\end{lemma}
\begin{IEEEproof}[Proof of Lemma \ref{pr:blind_identification}]
Follows immediately by invoking the definition of stable identifiability and noting that  
$\alpha \norm[\opclass]{H_1 -H_2} \leq \norm[]{ (H_1- H_2 )\prsig  } \leq \beta \norm[\opclass]{ H_1-H_2 }, \forall H_1,H_2 \in \opclassbr$ if and only if  $\alpha \norm[\opclass]{H } \leq \norm[]{ H \prsig  } \leq \beta \norm[\opclass]{ H }$   for all $H \in \mc  H_{M_\Phi \cup M_\Theta}$ and  for all $M_{\Phi}, M_{\Theta},$ with $\area(M_{\Phi}) \leq \SSF$ and $\area(M_{\Theta}) \leq  \SSF, \; \Phi,\Theta \in \Sigma$. 
\end{IEEEproof}

Before formally proving necessity in Theorem  \ref{th:sufficiency}, we comment on an important aspect of the difference between known and unknown spreading function support region. If the support region is known and given, by say $M_\Gamma$, and we consider the set $\opclassr$, it follows that the difference $H_1-H_2$ for $H_1,H_2 \in \opclassr$ satisfies $(H_1-H_2) \in \opclassr$. If we, however, consider the set $\opclassbr$, we have that $H_1,H_2 \in \opclassbr$ does not imply $(H_1-H_2) \in \opclassbr$ in general. To see this, simply take $H_1,H_2 \in \opclassbr$ such that the support regions of $H_1$ and $H_2$ have area $\SSF$ and are disjoint. We do, however, have that $H_1-H_2 \in \mathcal X(2 \SSF), \; \forall H_1, H_2 \in \opclassbr$. 
This observation lies at the heart of the factor of two penalty in $\SSF$ as quantified by Theorem \ref{th:sufficiency}.

\begin{IEEEproof}[Necessity in Theorem  \ref{th:sufficiency}]
Follows by taking $\SSF > \sfrac{1}{2}$ and noting that we can find 
$M_\Phi,M_\Theta$ with $\area(M_\Phi)=\area(M_\Theta)=\SSF$ and $M_\Phi \cap M_\Theta = \emptyset$. 
This implies $\area ( M_\Phi \cup M_\Theta  ) > 1$ and hence application of Theorem \ref{pr:necessity} to the set $\mc H_{M_\Phi \cup M_\Theta}$ establishes that $\mc H_{M_\Phi \cup M_\Theta}$ is not stably identifiable. By Lemma \ref{pr:blind_identification} this then implies that $\opclassbr$ is not stably identifiable. 
\end{IEEEproof}


\section{Sufficiency in Theorem \ref{th:sufficiency} \label{sec:suff}}

We prove sufficiency in Theorem \ref{th:sufficiency} by finding a probing signal that stably identifies $\opclassbr$. 
Concretely, we choose a weighted $\csamp L$-periodic train of Dirac impulses 
\begin{equation}
\prsig(t) = \sum_{k \in \mb Z}  c_k  \delta(t+   k\csamp ), \quad c_k = c_{k+L}, \; \forall k \in \mb Z
\label{eq:probingsignal}
\vspace{-0.1cm}
\end{equation}
as probing signal. The specific choice of the coefficients $\mbf c = \{c_0,...,c_{L-1}\}$ will turn out to be crucial and will be discussed later. 
The main idea of our proof is to 1)   
reduce the identification problem to that of solving a linear system of $L$ equations with $L^2$ unknowns, and 2) to apply Lemma \ref{pr:blind_identification} to show that a unique solution of this underdetermined system of equations exists whenever $\SSF \leq \sfrac{1}{2}$ (and $\mbf c$ is chosen appropriately). 

We start by computing the response of $H$ to $\prsig(t)$ in \eqref{eq:probingsignal}: 
\begin{align}
\vspace{-0.1cm}
y(t) 	&= (H \prsig) (t)  
	        =  \sum_{k \in \mb Z}  c_k \int_\nu   \sfunc(t+ k \csamp  ,\nu)  e^{j 2 \pi \nu t} d\nu. \label{eq:dirac}
\vspace{-0.1cm}	        
\end{align}
Next, define the Zak transform \cite{janssen_zak_1988} (with parameter $\csamp L$) of the signal $u(t)$ as 
\begin{equation}
\vspace{-0.2cm}
\mc Z_{u}(t ,f)   \nonumber \defeq \sum_{m \in \mb Z} u(t - m \csamp L ) e^{j2\pi  m \csamp L f} 
\vspace{-0.0cm}
\end{equation}
for $(t,f) \in  [0, \csamp L)  \times [0,1/\csamp L)$. 
The Zak transform of $y(t)$ in \eqref{eq:dirac} is given by
\vspace{-0.1cm}
\begin{align}
&\mc Z_{y} (t,f) =  \nonumber \\
&=\!\! \sum_{k,m \in \mb Z} \!\! c_k  \!\! \int_\nu \!  \sfunc(t - m \csamp L+k \csamp ,\nu)  e^{j 2 \pi \nu (t -m \csamp L)} d\nu \hspace{0.045cm} e^{j2\pi m \csamp L  f}  \nonumber \\
&= \! \! \sum_{k',m \in \mb Z}  c_{k'}   \int_\nu   \sfunc(t+ k' \csamp ,\nu)  e^{j 2 \pi \nu t}   e^{-j 2 \pi \nu m \csamp L} d\nu \hspace{0.01cm} \,e^{j2\pi m\csamp L f} \nonumber \\
&= \sum_{k \in \mb Z}  c_{k} \frac{1}{\csamp L} \sum_{m \in \mb Z} \sfunc \!  \left(t+ k \csamp, f + \frac{m}{\csamp L} \right)  e^{j 2 \pi t \left( f + \frac{m}{\csamp L} \right) } \nonumber 
\end{align}
where we used the substitution $k'= k-mL$ and the last step follows from the Poisson summation formula. 
Next, we substitute $t = t'+ p\csamp$ with $p \in \{0,...,L-1\}$ and $t' \in [0,\csamp)$. This amounts to splitting the fundamental rectangle $[0, \csamp L)  \times [0,1/\csamp L)$ of the Zak transform into $L$ ``cells'' $\Un$, where $\Un$  was defined in Section \ref{sec:prform}, and  yields
\begin{align*}
&z_p(t',f)  \defeq  \mc Z_{y}( t'+ p\csamp ,f), \quad (t',f) \in \Un,\quad p = 0,...,L-1 \\
&= \!\! \sum_{k,m \in \mb Z}  \! \frac{c_{k} }{\csamp L}   \sfunc \! \left(t'+p \csamp  + k \csamp  , f + \frac{m}{\csamp L}  \right) e^{j 2 \pi (t'+ p \csamp ) \left(f + \frac{m}{\csamp L} \right) } \\
&=  \sum_{k' \in \mb Z}  \frac{c_{k'-p} }{\csamp L}  \sum_{m \in \mb Z}  \sfunc \! \left(t'+ k' \csamp , f + \frac{m}{\csamp L}  \right) e^{j 2 \pi (t'+p \csamp) \left(f + \frac{m}{\csamp L} \right) } \\
&= \sum_{k =0}^{L-1} \frac{c_{k-p} }{\csamp L}   \sum_{m = 0}^{L-1} \sfunc \! \left( t'+ k \csamp, f + \frac{m}{\csamp L}  \right) e^{j 2 \pi (t'+ p \csamp) \left(f + \frac{m}{\csamp L}  \right) }
\end{align*}
where we used the substitution $k'=p+k$ and the last step is a consequence of $\sfunc(\tau , \nu ) = 0$ for $(\tau,\nu) \notin \left[0, \csamp L \right) \times \left[0, 1/\csamp \right)$, by definition. 
We can now rewrite the last equation in vector-matrix form by defining the column vectors $\mbf z(t,f)$ and $\mbf s(t,f)$: 
\[
[\mbf z(t,f)]_p \defeq  z_{p}(t,f)e^{-j2\pi p \csamp f} , \quad p=0,...,L-1
\]
and 
$\mbf s(t,f)  \defeq [s_{0,0}(t,f),s_{0,1}(t,f),\allowbreak  ...,\allowbreak  s_{0,L-1}(t,f), \allowbreak s_{1,0}(t,f),..., s_{L-1,L-1}(t,f)]^T$
with  
 \begin{equation}
 s_{k,m}(t,f) \defeq \sfunc \! \left( t + k  \csamp , f +   \frac{m}{\csamp L}\right)  e^{j 2 \pi \left( f + \frac{m}{\csamp L} \right)t}. 
 \label{def:skm}
\end{equation}
It is easily seen that the vector $\mbf s(t,f)$, $(t,f) \in \Un$, fully characterizes the spreading function. 
With all definitions in place, we finally obtain
\vspace{-0.2cm}
\begin{equation}
\mbf z(t,f) =   \mbf A_{\mbf c}  \mbf s(t,f), \quad  (t,f) \in  U
\label{eq:sysofeq}
\vspace{-0.2cm}
\end{equation}
with the $L\times L^2$ matrix
\vspace{-0.2cm}
\[
\mbf A_{\mbf c} \defeq  [\mbf A_{\mbf c, 0} | \; ...  \; | \mbf A_{\mbf c, L-1}  ], \quad
\mbf A_{\mbf c, k} \defeq  \frac{1}{ \csamp L }  \mbf C_{\mbf c, k} \herm{\mbf F}
\vspace{-0.2cm}
\]
where 
$[\mbf F]_{p,m} =e^{-j2\pi \frac{pm}{L}}, \; p,m=0,...,L-1$, and $\mbf C_{\mbf c, k}$ is the diagonal matrix with diagonal entries $\{c_{k}, c_{k-1}, ..., c_{k+1}\}$. 

The proof will be effected by applying Lemma \ref{pr:blind_identification}. 
Concretely we will prove stable identifiability of $\mc H_{M_\Phi \cup M_\Theta}$ for all pairs $M_\Phi,M_\Theta$ with 
$\area(M_\Phi) \leq \sfrac{1}{2}$ and $\area(M_\Theta) \leq \sfrac{1}{2}$. By setting $M_\Gamma = M_\Phi \cup M_\Theta$, this is equivalent to proving stable identifiability of $\mc H_{M_\Gamma}$ for all $M_\Gamma$ with $\area(M_\Gamma)\leq 1$. We therefore consider $H\in \opclassr$, and note that, by definition $s_{k,m}(t,f) = 0, \; \forall  (k,m) \notin \Gamma$.  Denote the vector obtained from $\mbf s(t,f)$ by selecting the entries corresponding to the active cells $\Gamma$ by $\mbf s_\Gamma(t,f)$ and let $\mbf A_\Gamma$ be the matrix containing the columns of  $\mbf A_{\mbf c}$ that correspond to these cells. Then \eqref{eq:sysofeq} becomes\footnote{
Pfander and Walnut \cite{pfander_measurement_2006} used the probing signal \eqref{eq:probingsignal} to prove that, for known spreading function support region, $\SSF \leq 1$ is sufficient for stable identifiability. The  crucial difference between \cite{pfander_measurement_2006} and our setup is that we need {\it each} submatrix of $\mbf A_{\mbf c}$ of $L$ columns to have full rank, as we do not assume knowledge of the support region. 
}
\vspace{-0.2cm}
\begin{equation}
\mbf z(t,f) =  \mbf A_{\Gamma}  \mbf s_\Gamma(t,f), \quad (t,f) \in  U.
\label{eq:sysofeq_rest}
\vspace{-0.2cm}
\end{equation}
Next, we formally relate \eqref{eq:sysofeq_rest} to the definition of stable identifiability through the following lemma, whose proof is given in the appendix. 
\begin{lemma}
Let $\prsig$ be given by \eqref{eq:probingsignal}. Then, the bounds $\alpha, \beta$ in \eqref{eq:cond_stable} for the set of operators $\opclassr$ are given as
\vspace{-0.1cm}
\begin{equation}
 \alpha_\Gamma = \sqrt{TL} \inf_{\norm[\ell_2]{\mbf v} = 1}( \mbf A_{\Gamma} \mbf v), \; \beta_\Gamma = \sqrt{TL} \sup_{\norm[\ell_2]{\mbf v} = 1}( \mbf A_{\Gamma} \mbf v).
\label{eq:stabbounds}
\end{equation}
\vspace{-0.1cm}
\label{le:boundedness_classr}
\end{lemma}
\vspace{-0.2cm}

The proof of sufficiency in Theorem \ref{th:sufficiency} is now completed by showing that for all $M_\Gamma$ with $\area(M_\Gamma)\leq 1$, $\opclassr$ is stably identifiable, i.e., $\alpha_\Gamma >0$. 
This amounts to proving that $\mbf A_\Gamma$ has full rank for all $M_\Gamma$ such that $\area(M_\Gamma)\leq 1$, i.e., for all $ \Gamma \in \Sigma$ such that $|\Gamma| \leq L$.  
What comes to our rescue here is a result in \cite{lawrence_linear_2005} which states that for almost all $\mbf c$, each $L\times L$ submatrix of $\mbf A_{\mbf c}$ has full rank. Hence, there exists a $\mbf c$ such that $\alpha_\Gamma > 0$ for all $M_\Gamma$ with $\area(M_\Gamma)\leq 1$. 
In the remainder of the paper, $\mbf c$ is chosen such that each $L\times L$ submatrix of $\mtx{A}_{\ve{c}}$ has full rank.

\section{Recovering the Spreading Function \label{sec:reconstruct} }
We next discuss an algorithm for the recovery of operators $H\in \opclassbr$ from the operator's response $H\prsig$ to the probing signal $\prsig(t)$ in \eqref{eq:probingsignal}. 
We start by noting that recovering $H$ amounts to recovering $\sfunc$ which by \eqref{def:skm} is equivalent to recovering $\mbf s(t,f)$ from \eqref{eq:sysofeq}. 
This will be accomplished by first identifying the support set of $\sfunc$, (i.e., the active cells of $\mbf s(t,f)$), and then solving the resulting linear system of equations \eqref{eq:sysofeq_rest}. 

\subsubsection{Support set recovery}

\newcommand{\IMV}[0]{(\text{P0})} 
\newcommand{\IMVM}[0]{(\text{$\overline{\text{P0}}$})} 
\newcommand{\MMV}[0]{(\text{P0$'$})}	

If we assume that $\area(M_\Gamma) \leq \sfrac{1}{2}$, then the support set $\Gamma$ can be recovered from $\mbf z(t,f)$ by solving: 
\[
\IMV \; \begin{cases} 
		\text{minimize} 	& |\Gamma| \\
		\text{subject to} 	& \mbf z(t,f) = \mbf A_{\Gamma}  \mbf s_\Gamma(t,f),\quad (t,f) \in \Un
\end{cases}
\]
where the constraint is over all $\mbf s_\Gamma(t,f)$ with $|\Gamma| \leq \sfrac{L}{2}$. 
This is a standard problem, and solutions have been proposed in the context of spectrum-blind sampling \cite{feng_spectrum-blind_1996,bresler_spectrum-blind_2008,mishali_blind_2009}, all involving a correlation matrix, which in our setup becomes $\mbf Z \defeq \int_{\Un} \mbf z(t,f) \herm{\mbf z}(t,f) d(t,f)$. The main difference to signal recovery in the context of spectrum-blind sampling  \cite{feng_spectrum-blind_1996,bresler_spectrum-blind_2008,mishali_blind_2009} is that here a function of two variables has to be recovered rather than a function of one variable. 
Using \eqref{eq:sysofeq_rest} we can express $\mbf Z$ according to
\begin{equation}
\mbf Z  =  \mbf A_{\Gamma} \mbf S_{\Gamma}  \herm{\mbf A}_{\Gamma}
\label{eq:integral}
\end{equation}
where $\mbf S_{\Gamma} = \int_{\Un}   \mbf s_{\Gamma}(t,f)  \herm{\mbf s}_{\Gamma}(t,f) d(t,f)$. 
Using similar arguments as in \cite{feng_universal_1997}, it can be shown that \IMV~is equivalent to
\[
\IMVM \; \begin{cases} 
		\text{minimize} 	& |\Gamma| \\
		\text{subject to} 	&  \mbf Z =  \mbf A_{\Gamma} \mbf S_{\Gamma}  \herm{\mbf A}_{\Gamma}
\end{cases}
\]
 where the constraint is over all Hermitian matrices $\mbf S_{\Gamma} \in \mb C^{|\Gamma| \times |\Gamma|}$. Since $\mbf Z$ is normal, it can be decomposed as $\mbf Z = \mbf Q \herm{\mbf Q}$, where the $R = \rank{\mbf Z}$ columns of the matrix $\mbf Q \in \mb C^{L \times R}$ are orthogonal. As shown by Feng \cite{feng_universal_1997}, \IMVM~(and hence \IMV) is equivalent to
\begin{equation}
\MMV \; \begin{cases} 
		\text{minimize} 	& |\Gamma| \\
		\text{subject to} 	& \mbf Q = \mbf A_{\Gamma}  \mbf G_{\Gamma} 
\end{cases}
\label{eq:constraintmmv}
\end{equation}
where the constraint is over all $\mbf G_{\Gamma} \in \mb C^{|\Gamma| \times R}$. \MMV~is known in the literature as the finite multiple-measurement vector (MMV)
 problem\footnote{The MMV problem is usually formulated as follows: Minimize $\norm[\text{row-0}]{\mtx{G}}$ subject to $\mtx Q=\mtx{A}_{\ve{c}} \mtx{G}$, where the constraint is over all $\mtx{G} \in \complexset^{L^2\times R}$ and $\norm[\text{row-0}]{\mtx{G}}$ is the number of rows of $\mtx{G}$ that contain at least one non-zero entry.} \cite{jie_chen_theoretical_2006}.  
Application of 
\cite[Thm. 2.4]{jie_chen_theoretical_2006} ensures that (P0$'$) provably recovers the correct support set $\Gamma$ as long as $|\Gamma| \leq \sfrac{L}{2}$, i.e., as long as the area of the unknown support set of the spreading function satisfies $\SSF \leq \sfrac{1}{2}$. 

\subsubsection{Recovery for known support set\label{sec:knownsupp}}
Once the support set $\Gamma$ has been identified, we solve \eqref{eq:sysofeq_rest} for $\mbf s_\Gamma(t,f)$, which based on \eqref{def:skm} yields $\sfunc$ and hence $H$. 
Note that \eqref{eq:sysofeq_rest} has to be solved over the continuum of values $(t,f) \in \Un$. 
We can expand all quantities in  \eqref{eq:sysofeq_rest} into two-dimensional Fourier series over $\Un$, which results in a system of countably many linear equations to be solved. 

\section{Identification for Almost All $H \in \opclassbr$}

We next prove Theorem \ref{th:almost_all}. 
The proof is inspired by \cite[Thm. 1]{bresler_spectrum-blind_2008}, and is constructive as it specifies the recovery algorithm (for almost all $H\in \opclassbr, \SSF < 1$). 
The basic idea is to use a MUSIC-like \cite{schmidt_multiple_1986} algorithm based on \eqref{eq:integral}, which allows us to recover $\Gamma$ under the following two conditions: 
\begin{enumerate}
\item \label{cond_1} $\SSF \leq 1-\sfrac{1}{L}$.  The penalty of $\sfrac{1}{L}$ is technical and can be made arbitrarily small by choosing $L$ large enough. 
\item \label{cond_2} The functions  $s_{k,m}(t,f), (k,m) \in \Gamma$,  are \emph{linearly independent} on $\Un$, i.e., there is no vector $\mbf a \in \mb C^{N}, \mbf a \neq \mbf 0, N=|\Gamma|$, such that $\herm{\ve{a}} \mbf s_\Gamma(t,f) =0, \,\, \forall \; (t,f) \in \Un$.
\end{enumerate}
We recognize that almost all  $H\in \opclassbr$ satisfy Condition \ref{cond_2}. 

\begin{proof}[Proof of sufficiency in Theorem \ref{th:almost_all}]
The proof is effected by establishing that, under Conditions \ref{cond_1} and \ref{cond_2} above, the support set $\Gamma$ is uniquely specified by the indices of the columns of $\herm{\mbf U}_n \mbf A_{\mbf c}$ that are equal to $\mbf 0$. Here, $\mbf U_n$ is the matrix of eigenvectors of $\mbf Z$ corresponding to zero eigenvalues. 
To see this, we perform an eigenvalue decomposition of $\mbf Z$ in \eqref{eq:integral} to obtain
\begin{equation}
\mbf Z = 
\begin{bmatrix}
\mbf U_z  & \mbf U_n 
\end{bmatrix}
  \begin{bmatrix} \mbf \Lambda_z & \mbf 0\\   \mbf 0 &  \mbf 0   \end{bmatrix}  
  \begin{bmatrix}
  \herm{\mbf U}_z \\
  \herm{\mbf U}_n
  \end{bmatrix}
  = \mbf A_{\Gamma}  \mbf S_\Gamma \herm{\mbf A}_{\Gamma}
\label{eq:ev1}
\end{equation}
where $\mbf U_{z}$ contains the eigenvectors of $\mbf Z$, corresponding to the non-zero eigenvalues of $\mbf Z$. 
Lemma \ref{pr:fullrankz} in the appendix establishes that, thanks to Condition \ref{cond_2} above, for almost all $H\in \opclassbr$, $\ve{S}_{\Gamma}$ has full rank. 
As discussed in Section \ref{sec:suff}, each set of $L$ or fewer columns of $\mbf A_{\mbf c}$ is linearly independent. 
Condition \ref{cond_1} ensures that $|\Gamma|\leq L-1$, which implies that $\mbf A_\Gamma$ has full rank for all sets $\Gamma$ in question. 
Hence we get from 
 \eqref{eq:ev1} that 
 \begin{equation}
 \range(\mbf A_{\Gamma}) = \range(\mbf A_{\Gamma}  \mbf S_\Gamma \herm{\mbf A}_{\Gamma})= \range(\mbf U_z \mbf \Lambda_z \herm{\mbf U}_{z}) = \range(\mbf U_z).
\label{eq:rankeq}
\end{equation}
 $\range(\mbf U_n)$ is the orthogonal complement of $\range(\mbf U_z)$ in ${\mathbb C}^L$. It therefore follows from \eqref{eq:rankeq} that $\herm{\mbf U}_n \mbf A_\Gamma=\mbf 0$. 
Therefore, the columns of $\herm{\mbf U}_n \mbf A_{\mbf c}$ that correspond to indices $(k,m) \in \Gamma$ are equal to $\mbf 0$. 
To conclude the proof, we show, by contradiction, that no other columns of $\herm{\mbf U}_n \mbf A_{\mbf c}$ are equal to $\mbf 0$. 
Suppose that $\herm{\mbf U}_n \mbf a = \mbf 0$ where $\mbf a$  is any column of $\mbf A_{\mbf c}$ corresponding to an index pair $(k',m') \notin \Gamma$. 
Then $\mbf a \in \range(\mbf U_z)= \range(\mbf A_{\Gamma})$. 
This would mean that the $L$ or fewer columns of $\mbf A_{\mbf c}$ corresponding to the indices $(k,m) \in \{\Gamma \cup (k',m')\}$ would be linearly dependent. The proof is completed by noting that this stands in contradiction to the fact that---since we assume that $\mbf c$ is chosen accordingly---each set of $L$ or fewer columns of $\mbf A_{\mbf c}$ is linearly independent. 
\end{proof}

\appendix

\begin{IEEEproof}[Proof of Lemma \ref{le:boundedness_classr}]
Starting from \eqref{eq:sysofeq_rest}, we get for fixed values of $(t,f)\in \Un$
\begin{equation}
 \frac{\alpha_\Gamma}{\sqrt{TL}}  \norm[\ell_2]{\mbf s_\Gamma (t,f)}   \leq   \norm[\ell_2]{\mbf z (t,f)} \leq   \frac{\beta_\Gamma}{\sqrt{\csamp L}}  \norm[\ell_2]{\mbf s_\Gamma (t,f)} 
 \label{eq:upperlower}
\end{equation}
 with $\alpha_\Gamma, \beta_\Gamma$ defined in \eqref{eq:stabbounds}. 
 Squaring and integrating \eqref{eq:upperlower} over $\Un$ yields 
\newpage
\phantom{adf}
\vspace{-0.9cm}
\begin{align}
&\int_{\Un}  \norm[\ell_2]{\mbf z (t,f)}^2 d(t,f) = \sum_{p = 0}^{L-1} \int_{\Un}  \left| z_p(t,f)  \right|^2  d(t,f)  \nonumber \\
&= \int_{ [0, \csamp L)  \times [0,1/\csamp L)}  |\mc Z_{y} (t,f) |^2 d(t,f) =  \frac{1}{\csamp L} \norm[]{H \prsig}^2
\label{eq:hzeq}
\end{align}
where the last equality follows from the unitarity of the Zak transform \cite{janssen_zak_1988}. 
Similarly, we have
\vspace{-0.1cm} 
\begin{equation}
\int_{\Un}   \norm[\ell_2]{\mbf s_\Gamma (t,f)}^2  d(t,f)  = \norm[]{\sfunc}^2 = \norm[\opclass]{H}^2.
\label{eq:normheqints}
\vspace{-0.1cm}
\end{equation}
Inserting \eqref{eq:normheqints} and \eqref{eq:hzeq}  into \eqref{eq:upperlower} completes the proof. 
\end{IEEEproof}

\begin{lemma}
$\ve{S}_{\Gamma}$ has full rank if and only if the functions $s_{k,m}(t,f), (k,m) \in \Gamma$, are linearly independent on $\Un$.   
		  \label{pr:fullrankz}
\end{lemma} 
\begin{IEEEproof}
Assume that  $\ve{S}_{\Gamma}$ does not have full rank. Then there exists an $\ve{a} \neq \mbf 0, \ve{a} \in \mb C^{N}$ such that $\herm{\ve{a}} \ve{S}_{\Gamma} = \mbf 0$ and hence $\herm{\ve{a}} \ve{S}_{\Gamma} \ve{a} = \int_{\Un} 
\herm{\ve{a}} \mbf s_{\Gamma}(t,f) \hspace{0.015cm} \herm{( \herm{\ve{a}} \mbf s_{\Gamma}(t,f) )} d(t,f) = 0$. Since $ \herm{\ve{a}} \mbf s_{\Gamma}(t,f)  \herm{( \herm{\ve{a}} \mbf s_{\Gamma}(t,f) )} \geq 0, \forall (t,f) \in  \Un$, we must have $\herm{\ve{a}} \mbf s_{\Gamma}(t,f) = 0$ a.e. on $\Un$, which implies that the set $s_{k,m}(t,f),   (k,m) \in \Gamma$, is linearly dependent on $\Un$. 
Now assume that the set $s_{k,m}(t,f),   (k,m) \in \Gamma$, is linearly dependent on $\Un$. Then, there exists an $\ve{a}\neq \mbf 0$ such that  $\int_{\Un} \herm{\ve{a}} \mbf s_{\Gamma}(t,f) \hspace{0.015cm} \herm{\mbf s}_{\Gamma}(t,f) d(t,f) =\mbf 0$, and hence, using $\ve{S}_{\Gamma} = \int_{\Un} \mbf s_{\Gamma}(t,f) \hspace{0.015cm} \herm{\mbf s}_{\Gamma}(t,f) d(t,f)$, by definition, we get 
$\herm{\ve{a}} \ve{S}_{\Gamma}= \mbf 0$ which proves that $\ve{S}_{\Gamma}$ does not have full rank.
	\end{IEEEproof}

\vspace{-0.1cm}

\renewcommand{\baselinestretch}{0.943}\small\normalsize


\end{document}